\documentclass[twocolumn,superscriptaddress,preprintnumbers,amsmath,amssymb]{revtex4}

\usepackage{graphicx}
\usepackage{dcolumn}
\usepackage{bm}
\usepackage{amssymb,amsmath,amsthm,amsfonts}
\usepackage{bbm}
\usepackage{mathtools}
\usepackage{enumitem}
\usepackage{xcolor}
\usepackage{textcomp}
\usepackage{gensymb} 
\usepackage{float}
\usepackage{soul}
\usepackage{booktabs}
\usepackage{breakurl}

\newtheorem{theorem}{Theorem}[section]
\newtheorem{corollary}{Corollary}[theorem]

\newcommand{\hush}[1]{}

\begin{document}
\title{Enhancing quantum computer performance via symmetrization}
\author{Andrii Maksymov}
\affiliation{IonQ, College Park, MD 20740, USA}
\author{Jason Nguyen}
\affiliation{IonQ, College Park, MD 20740, USA}
\author{Yunseong Nam}
\affiliation{IonQ, College Park, MD 20740, USA}
\affiliation{Department of Physics, University of Maryland, College Park, MD 20742, USA}
\author{Igor Markov}
\affiliation{IonQ, College Park, MD 20740, USA}
\date{\today}

\begin{abstract}
Large quantum computers promise to solve some critical problems not solvable otherwise.
However, modern quantum technologies suffer various imperfections such as control errors and qubit decoherence, inhibiting their potential utility.
The overheads of quantum error correction are too great for near-term quantum computers, whereas error-mitigation strategies that address specific device imperfections may lose relevance as devices improve.
To enhance the performance of quantum computers with high-quality qubits, we introduce a strategy based on \emph{symmetrization} and nonlinear aggregation. On a commercial trapped-ion quantum computer, it improves performance of multiple practical algorithms by 100x with no qubit or gate overhead.
\end{abstract}

\maketitle

\section{Introduction}

Quantum computers (QCs) are rapidly growing in capacity, but are held back by quantum noise, decoherence, crosstalk and gate control inaccuracies~\cite{Erhard2019,Wright_2019,Cian2021CrossPlatformCO,SupermarQ2022}. 
Each qubit technology seeks to suppress such irregularities for individual qubits and gates~\cite{Maksymov:2021brv,blumel2019poweroptimal,blumel2021efficient,Ballance2016,Elder2020,isca2021}. However, the circuit fidelity provided by these methods falls short by orders of magnitude compared to the needs of large-scale quantum algorithms.
This necessitates the development of \emph{higher-level} strategies that systematically improve performance as observed at the algorithmic level, and we offer such techniques in our work.
As in conventional computers, firmware in quantum computers  
provides necessary low-level control for a device's specific hardware and orchestrates hardware so that software can run more effectively and efficiently. 
Firmware can implement quantum error correcting codes (QECC) that mathematically promise to tolerate small-enough irregularities via wide-circuit redundancy. 
However, for near-term quantum computers, irregularities are often too great for these codes to function properly~\cite{NC}. 
Leading QECC techniques require many additional qubits, gates, measurements, low-latency classical-control interconnects, and exorbitant amounts of supporting nonquantum computation. 
Although eventually QECC promises attractive scalibility, present-day quantum computers are far too small to benefit from QECC and wide-circuit redundancy~\footnote{The overhead for QECC is typically 5-7x the number of qubits. For NISQ system with 50-70 qubits this leaves very few logical qubits for quantum computation.}.

To make progress with present-day QCs, researchers have developed alternate firmware approaches known as \emph{error mitigation}. In leading superconducting QCs, the quality of individual physical qubits varies enough for the result to depend on the mapping of logical to physical qubits.
Therefore, researchers try to optimally map qubits~\cite{murali2019noiseadaptive,10.1145/3297858.3304007,Maksymov_2021}, order gates~\cite{10.1145/3386162}, and ensemble-average over circuit mappings to mitigate the effect of correlated errors with minimal overhead~\cite{10.1145/3352460.3358257}.
A series of techniques is based on first accurately characterizing quantum device irregularities and errors, then suppressing them by adjusting control pulses~\cite{Sun_2021}, probabilistically canceling them via applying extra gates~\cite{strikis2021learningbased,suzuki2021quantum,Temme_2017,piveteau2021error}, or using machine learning on the quantum computational output~\cite{strikis2021learningbased,czarnik2021error}. 
Another insight is that decorrelated noise accumulates at a smaller rate with the number of gates. Hence, \emph{gate-level decorrelation}~\cite{Campbell_2017,Kern_2005,Wallman2016,Cai_2020,hashim2020randomized} adds gates to decorrelate noise at the cost of some overhead, which can also add to the noise if significant. 
Researchers have used this effect to systematically \emph{amplify} noise, which allows one to extrapolate output states to the zero-noise limit~\cite{cai2021multiexponential,Temme_2017,Kandala_2019,Endo_2018,Songeaaw5686}.

Leading error-mitigation strategies developed and deployed for superconducting QCs address stochastic noise and uneven quality of physical qubits. To improve, superconduting QCs must attain uniformly-high qubit quality and low stochastic noise. After such improvements, the error-mitigation techniques we reviewed above may lose relevance.
Such improved technologies can be illustrated by the present-day trapped-ion QCs where practically-identical qubits enjoy long decoherence times and low random noise~\cite{Wang_2021,PhysRevA.94.042308,Kielpinski2002}.
The remaining adverse effects are due to slowly-drifting control inaccuracies~\cite{Maksymov:2021brv}.
In this work, we develop and validate novel error mitigation techniques for ion-trap QCs with expectation of broader applicability to present-day and future QCs. 

We introduce a firmware-level error mitigation strategy called \emph{symmetrization}.
To avoid qubit- and gate-level overhead, it distinguishes the ideal quantum computation by its invariance under certain symmetries that arise at multiple levels of physical implementation~\footnote{Such as qubit mappings, circuit compilation, gate decomposition, pulse sequences etc.}.
Our strategy first uses symmetries to generate variant circuit implementations.
These variants run on one or multiple QCs, and collected measurement statistics are aggregated via linear or nonlinear techniques.
Subsequently, symmetrized effects of deterministic inaccuracies largely cancel out while random noise does not get amplified.

We validated our strategy on the IonQ Aria commercial QC for quantum algorithms of practical interest~\cite{qedcArxiv,Daiwei_QML,IonQ_Q1_2022,IonQ_AQ20_2022}.
For quantum ML (QML) circuits~\cite{Daiwei_QML}, linear aggregation gives a 1.5-2$\times$ performance boost.
Nonlinear aggregation by voting provides much greater gains but may distort results if used inappropriately.
For a 15-qubit quantum Fourier transform (QFT) adder circuit~\cite{qedcArxiv} with voting, we see a 100$\times$ performance gain without distortion.
We explore the choice of aggregation in Section~\ref{sec:error-analysis} and provide a guide for future uses in Discussions.

\section{Results}
\subsection{Symmetrization strategy}

\begin{figure}
\includegraphics[scale=0.46]{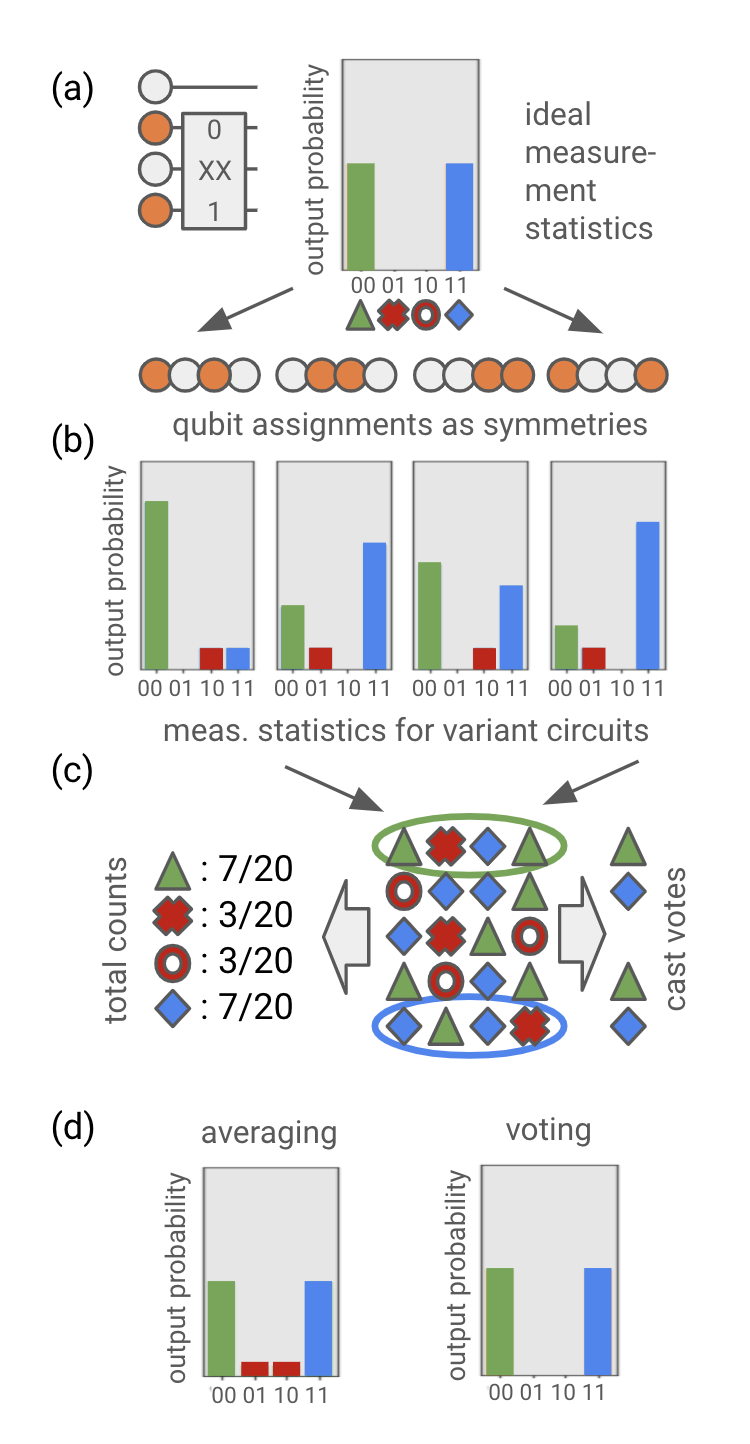}
\caption{\label{fig:flow} Symmetrized circuit execution: (a) splitting execution into symmetrized variants illustrated by varying qubit assignments, (b) measuring each result affected by individual inaccuracies, (c) aggregating measurement statistics while (d) compares the difference between averaged results and obtained through plurality voting. First, for
each of four selected qubit pairs, a circuit variant produces a superposition state $(\left|00\right>+\left|11\right>)/\sqrt{2}$  (target qubits are marked orange) (a). For each qubit pair, the output state is
recreated and measured five times in the computational basis (b). If the measurements are grouped per mapping, their statistics significantly deviate from the ideal, but approach the ideal when averaged; residual erroneous counts are shown as red circles and crosses, while all-zero states are green triangles, all-one are blue diamonds (c). When aggregated with plurality vote taken across variants, erroneous counts are filtered out, whereas componentwise averaging preserves all counts (d).}
\end{figure}

\begin{figure}
\includegraphics[scale=0.41]{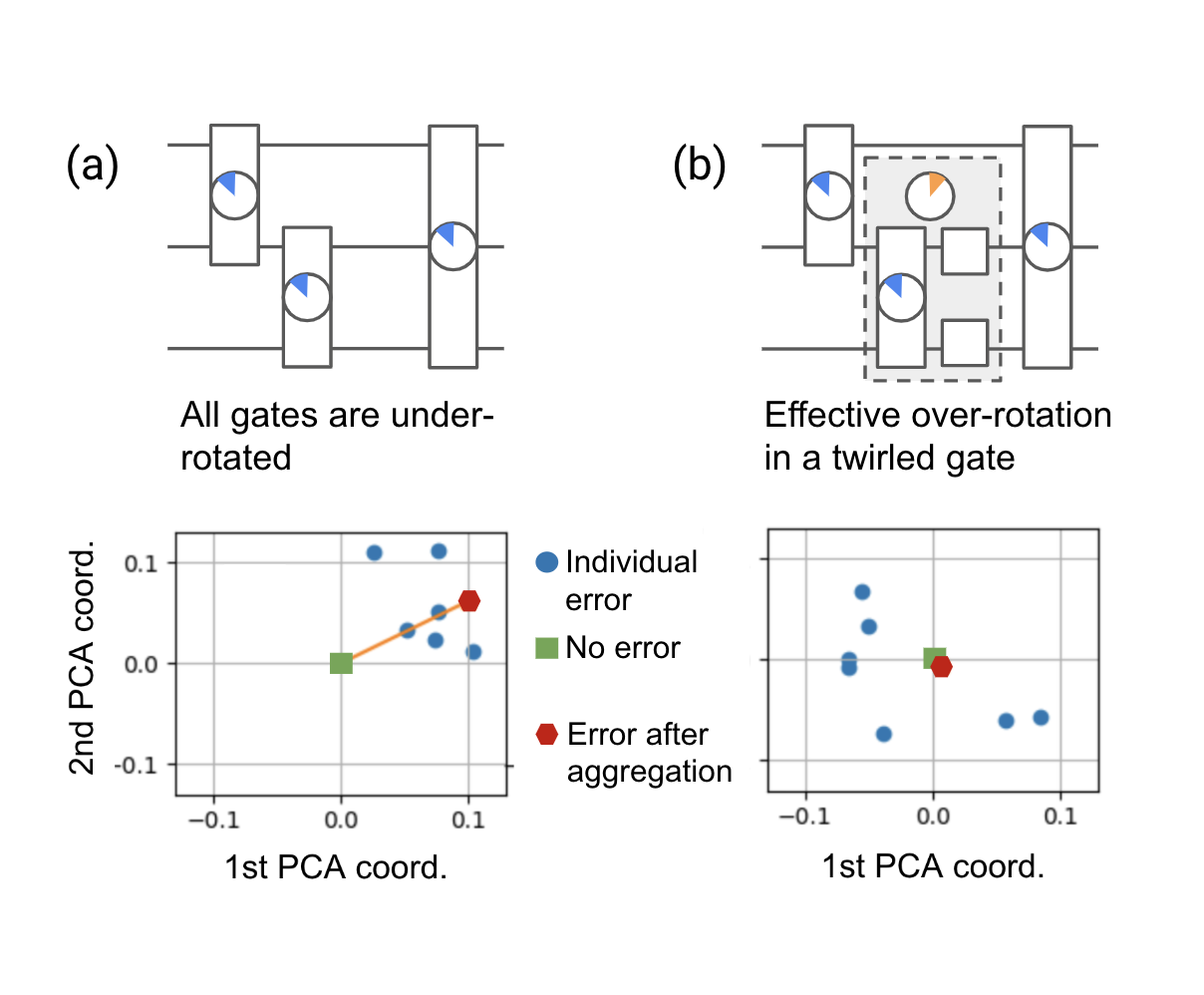}
\caption{\label{fig:pca}
Simulation results for symmetrized 4-qubit circuits (on eight ions) with average under-rotation on all qubit pairs. Panel (a) illustrates the use of qubit remappings as symmetries, while panel (b) shows combined use of qubit remapping and gate twirling. At the top, we contrast similar quantum circuit blocks. Circles with blue sectors mark gates with under-rotation while orange sectors mark gates with over-rotation. At the bottom, we plot the errors after aggregation in the space of the first two principal components of the deviation from the ideal histogram for biased (a) or symmetrized (b) inaccuracies. Individual output errors are shown for simulations of eight mappings of a random circuit on eight qubits with under-rotations specified per two-qubit gate.
}
\end{figure}

We consider a set of $n$-qubit computations $\mathfrak{U}_n$, each including state initialization, some operator from $SU(2^n)$, and final measurements. 
Let $\mathfrak{R}_n$ be a set of realizations (of all quantum computations in $\mathfrak{U}_n$) that represent gate-level quantum circuits with qubit assignment, initialization, measurements, postprocessing, and possibly implementation details such as pulse sequences specified. 
We define the function $\pi : \mathfrak{R}_n \to \mathfrak{U}_n$
that finds the computation $u$ performed by a given concrete realization $r$. 
We define the \emph{general symmetries of $\mathfrak{U}_n$}, denoted $\Gamma(\mathfrak{U}_n)$, as the set of functions $\gamma: \mathfrak{R}_n \to \mathfrak{R}_n$ that satisfy
\begin{align}
 \pi \circ \gamma = \pi.
\end{align}
That is, $\gamma \in \Gamma(\mathfrak{U}_n)$ if and only if for all $r \in \mathfrak{R}_n$, whenever $\pi(r) = u$, $\pi(\gamma(r)) = u$.
In other words, applying $\gamma$ to any realization will produce the same quantum computation.
We also define \emph{computation-specific symmetries}, $\Gamma(\mathfrak{U}_n, r)$, as the set of functions $\gamma: \mathfrak{R}_n \to \mathfrak{R}_n$ that satisfy $\pi(\gamma(r)) = \pi(r)$ for a particular $r$. 
For example, general symmetries could be conjugations (in group-action sense) of gate-level circuits by qubit permutations. Namely, the initial state is replaced by its permutation, the gates are applied on permuted qubits, and the measurement results are permuted back. Examples of computation-specific symmetries are  gate decompositions, permutations of commuting gates, the addition of gates that preserve a given state (e.g., before measurement), and changes of gates and measurements compensated by changes in postprocessing. When $\mathfrak{R}_n$ specifies pulse sequences, symmetries can replace them with physically equivalent ones.

By distributing the computation over multiple symmetries, we cancel out the effect of control inaccuracies without amplification of random errors. The steps of the procedure, as shown in Fig.~\ref{fig:flow}, are then:

\begin{enumerate}
\item Define symmetries $\Gamma$ and sample $\Gamma^\prime\subset\Gamma$.
\item Generate circuit variants for $\Gamma^\prime$.
\item Execute each variant on the QC hardware.
\item Aggregate all measurement statistics.
\end{enumerate}

\subsection{Choice of symmetries and aggregations}
\label{sec:error-analysis}
We now consider why symmetrization works. Let {\em inaccurate realizations} $\tilde{r}_u$ be determined by instantaneous parameters of the physical system, such that $\pi(\tilde{r}_u) = u+\delta u$. A key example is unitary under- or over-rotations of particular gates~\cite{Maksymov:2021brv}.

To mitigate the impact of inaccuracies,
we consider $\gamma(\tilde{r}_u)$ for multiple $\gamma \in  \Gamma =\Gamma(\mathfrak{U}_n, r_u)$ 
so as to symmetrize the error term in
$\pi(\gamma(\tilde{r}_u)) = u+\delta u_{\Gamma \text{-inv}}+\delta u_{\gamma}$.
In the absence of errors, all realizations $r_u$ of $u$ implement the same computation. In the presence of errors, we rely on symmetrization
over multiple $\gamma \in \Gamma$ to produce a computation $u +\delta u_{\Gamma\text{-inv}}+\left< \delta u_\gamma \right>_{\Gamma}$. As long as we select an uncorrelated set of $\gamma$, $||\left< \delta u_\gamma \right>_{\Gamma}|| \ll \left< ||\delta u_\gamma ||\right>_{\Gamma}$, and the cumulative effect of non-$\Gamma$-invariant errors is much reduced.

In practice, rather than aggregating all $\pi(\gamma(\tilde{r}_u))$, we consider the output states produced by $\pi(\gamma(\tilde{r}_u))$ and aggregate their measurement statistics (because, e.g., coherently adding two quantum states would require additional qubits). 
The impact of inaccuracies on an ideal distribution 
$\mathbf{h_u} \in \mathbb{R}^+(2^n)$ may be expressed as $\mathbf{h_u} + \delta h_{\Gamma\text{-inv}} +\delta h_\gamma$.
Aggregating measurement statistics can ``enhance the contrast'' between the target output states and erroneously observed states. The error terms may cancel out, but more typically they would be uncorrelated.
For example, if  $\mathbf{h_u}=(0,\dots,1_k,\dots,0_{2^n})$ and $\Gamma=S_{2^n}$, then the symmetrized result would be 

\begin{multline}
\label{eq:singleout}
\mathbf{h_u}+\delta h_{\Gamma\text{-inv}}+\left< \delta h_\gamma \right>_{\Gamma} = \mathbf{h_u}+\delta h_{\Gamma\text{-inv}} = \\ = \left(\frac{\varepsilon}{2^n-1},\dots,(1-\varepsilon)_k,\dots,\frac{\varepsilon}{2^n-1}\right)
\end{multline}

\noindent
where $\varepsilon$ is the average error on $k$ and $\Gamma$ is sufficiently large. The probability of output $k$ is no better, but other probabilities (that should ideally be 0) become less pronounced. This decreases the probability that an erroneous output is observed repeatedly by chance and helps find the desired outputs with fewer samples.

The term $\delta h_{\Gamma\text{-inv}}$ in Eq.~\ref{eq:singleout} captures the remaining fully-depolarizing error channel, i.e., the effect of incoherent errors~\cite{Takagi2022}. This residual error can be reduced with aggregation techniques such as \textit{plurality voting}, e.g., for $\mathbf{h_u}$ with $l$ output states of frequency $\frac{1}{l}$ if $\varepsilon<1-l/2^n$ as proven in Supplemental Materials.

As a concrete example, we demonstrate the effect of symmetrization on 4-qubit circuits with six two-qubit gates on different qubit pairs and random single-qubit gates mapped to eight ions (see Methods). We model gate miscalibrations as random under-rotations of multiple two-qubit gates fixed per qubit pair. We assume an average under-rotation across all qubit pairs causing a similar error for all variants (Fig.~\ref{fig:pca}a). Symmetries $\Gamma$ are represented by eight random qubit assignments $\gamma$. For each qubit assignment, we simulated corresponding inaccurate realizations to obtain vectors $\mathbf{h_u}+\delta h_{\Gamma\text{-inv}}+\delta h_\gamma$. In Fig.~\ref{fig:pca}, we illustrate 256-dimensional vectors for ideal, individual, and symmetrized results by plotting their two largest principal components (principal component analysis (PCA) was initially performed on $\{ \delta h_{\Gamma\text{-inv}}+\delta h_\gamma \}$ vectors).
In the first case, since all gates are under-rotated by some amount on average, the variants fail to symmetrize the errors because they are only exploring qubit assignment symmetries. Hence, error effects $\delta h_{\Gamma\text{-inv}}$ remain after aggregation as shown in Fig.~\ref{fig:pca}a. In the second example (Fig.~\ref{fig:pca}b), we use additional symmetries of gate decompositions, to zero out $\delta h_{\Gamma\text{-inv}}$. The effect of under-rotation in fully-entangling XX gates is addressed using an alternative implementation that combines phase-flipped XX$^{-1}$ gates with pairs of X-gates thus implementing the same ideal unitary but reversing the effect of under-rotation.

An aggregation strategy for measurement statistics is a procedure that combines measurement statics from multiple implementations of the same computation. Without errors, all implementations should produce identical statistics in the limit (with infinite repetition count).
An aggregation strategy is considered \textit{stable} for a given type of statistics if, provided a set of identical statistics of this type, it produces another copy. Aggregation by \textit{componentwise averaging} is trivially stable for statistics of any type. Yet aggregation by \textit{voting} is not. This can be seen for the probability distribution $(1-\varepsilon, \varepsilon)$ which voting-based aggregation brings closer to $(1,0)$ for $\varepsilon<1/2$. What makes aggregation strategies useful is that ($i$) they coerce arbitrary statistics to statistics of the desired type, ($ii$) they distill original statistics from 
multiple erroneous variants of the original. To this end, output probability distributions are analytically characterized for many quantum algorithms including Shor's and Grover's. The choice of aggregation is determined by the type of output probability distribution of a given quantum algorithm. 

For best performance, we recommend aggregation by plurality voting for quantum algorithms with ideal measurement statistics comprising of $l$ outputs with frequencies $\frac{1}{l}$. Such algorithms have zero-frequency outputs and a subset of target outputs that needs to be determined. For algorithms with different measurement statistics, aggregation by averaging can be used to avoid distortion.

\subsection{Experiment}

\label{sec:Experiment}

We evaluate the impact of symmetrization and the choice of aggregation strategy experimentally by comparing the results of unsymmetrized runs to symmetrized runs with componentwise averaging and plurality voting.
We use the IonQ Aria trapped-ion quantum computer for these experiments, configured to utilize 20 addressable qubits. See methods for experimental details.

Performance is measured by Hellinger fidelity, defined as a statistical overlap $F_H = \left(\sum_i \sqrt{p_i q_i}\right)^2$ between the actual output statistics $p_i$ and the ideal result $q_i$ is computed via an error-free simulator. $F_H$ ranges from $0$ to $1$, with 0 capturing probability distributions that do not overlap, and 1 corresponding to a pair of identical distributions.
Also known as the \textit{Bhattacharyya coefficient} \cite{bhattacharyya1943measure}, $F_H$ is commonly used to measure the discrepancy between probability distributions and is consistent with the definition of fidelity for quantum states.

\begin{figure}[t]
\includegraphics[scale=0.42]{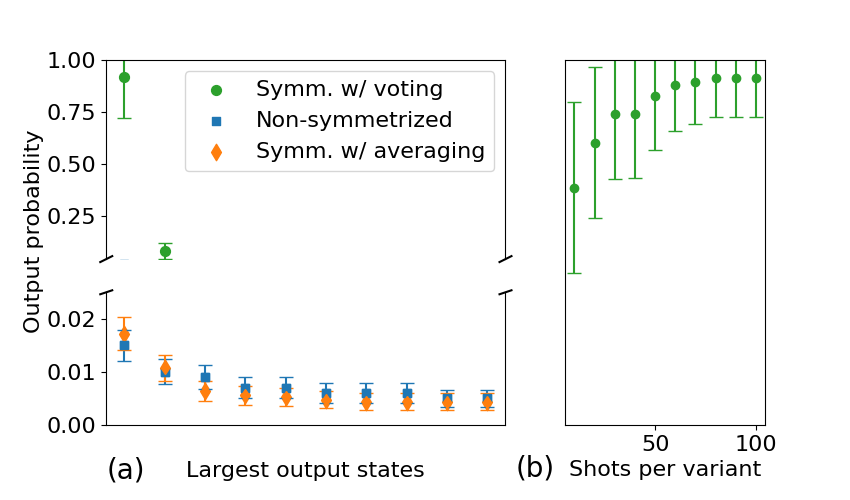}
\caption{\label{fig:improvement} Symmetrization of a 13-qubit single-output QFT-based adder circuit~\cite{qedcArxiv} boosts success probability when aggregated with plurality voting. (a) We compare the results without symmetrization and with symmetrization using either componentwise averaging or plurality voting. 
Blue squares show the unsymmetrized results, using a single realization with 2500 repeated measurements (the total number of experiments is the same in all three cases). Orange diamonds represent the symmetrization of this execution with 25 realizations and 100 repetitions per variant, aggregated with componentwise averaging. Green circles use the same realizations as orange points, but the symmetrized histogram is generated with plurality voting. This boosts the probability of the target outcome because the outcomes not matched between the variants are filtered out. In this case, the symmetrized results keep improving up until around 80 repetitions per variant (b).}
\end{figure}

\begin{figure*}
\includegraphics[scale=0.33]{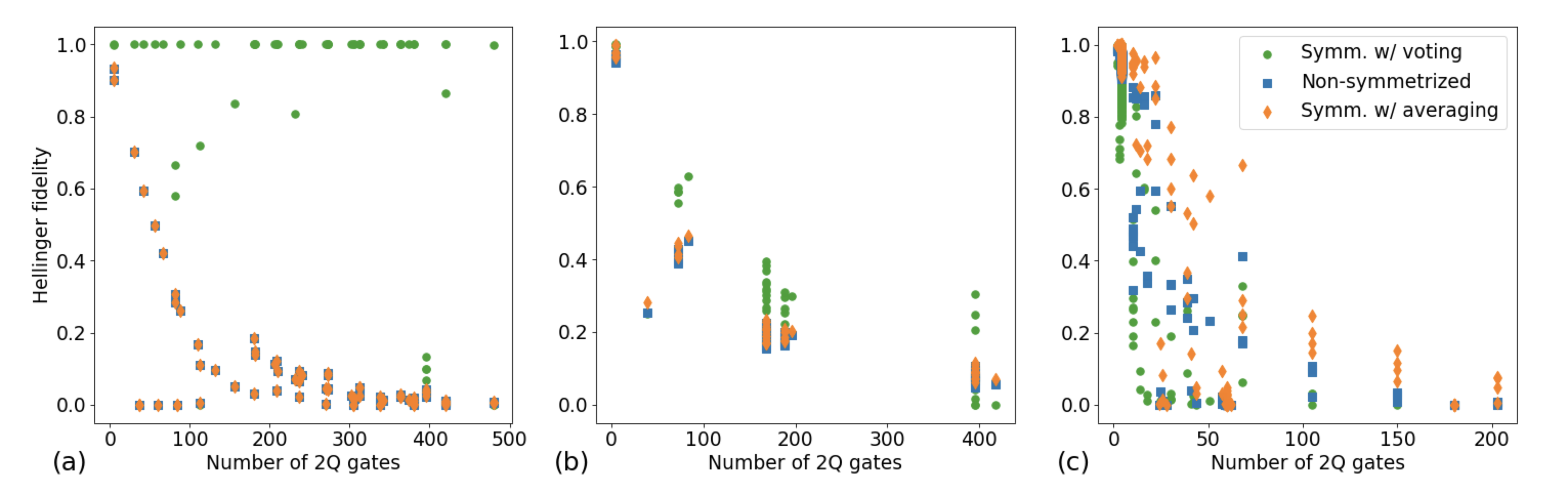}
\caption{\label{fig:fid} Comparison of fidelity improvement for algorithms (a) with one output state (quantum Fourier transform-based adders, phase estimation, and amplitude estimation), (b) with four output states (amplitude estimation and Monte Carlo sampling) or (c) with multiple output states (variational quantum eigensolvers and quantum machine learning), with and without symmetrization. Hellinger fidelity (see main text) is shown as a function of circuit depth, expressed in the number of two-qubit (2Q) MS (M{\o}lmer-S{\o}rensen) gates. Unsymmetrized results (green circles) are compared with results symmetrized and aggregated with plurality voting (blue squares) and componentwise averaging (orange diamonds). Unsymmetrized and symmetrized results are shown for the same set of experiments each consisting of 25 realizations with 100 repetitions per variant. Unsymmetrized fidelities are calculated as averages over 25 individual fidelities of each variant, which is why for the algorithms with one output state, they match exactly with the symmetrized results aggregated with componentwise averaging (a).}
\end{figure*}

We demonstrate the impact of symmetrization on a 13-qubit 
single-output QFT-based adder circuit~\cite{qedcArxiv}. In Fig.~\ref{fig:improvement}a we compare the largest output probabilities out of $2^{13}$ between the unsymmetrized histogram, symmetrized with componentwise averaging, and symmetrized with plurality voting. The first and largest value corresponds to the bitstring with ideal probability 1 while the rest should have otuput probability 0. Fig.~\ref{fig:improvement}b shows the change in the error bars with the number of shots.
We observe that symmetrization with componentwise averaging does not improve the probability of the target bitstring but does reduce next-largest probabilities, which allows for a dramatic increase in the probability of the target state after plurality voting (from 1.5\% to 95\%). For a 15-qubit QFT-based adder, the boost exceeds 100$\times$. 

Next, we examine the performance of symmetrization for several use cases shown in Fig.~\ref{fig:fid}. All jobs had 2500 shots taken with output probability distributions that vary in the number of correct output states with nonzero probability, and thus benefit differently from different aggregation strategies.
In Fig.~\ref{fig:fid}a, we evaluate results for QFT-based adders, phase estimation, and amplitude estimation with a single output state~\cite{qedcArxiv}. We see that symmetrization with plurality voting significantly increases $F_H$ while 
symmetrized runs with componentwise averaging show no improvement.
In Fig.~\ref{fig:fid}b, we compare results for amplitude estimation and Monte Carlo sampling circuits before tracing out the ancillary qubits. Symmetrization with plurality voting still shows the strongest improvement in $F_H$ but componentwise averaging is also better than no symmetrization because it evens out the errors across the four target states.
In Fig.~\ref{fig:fid}c, we evaluate symmetrization on variational quantum eigensolver (VQE) and quantum machine learning (QML) circuits~\cite{Daiwei_QML}. Those circuits have broader, irregular output distribution, so that symmetrization with componentwise averaging shows the best improvement while plurality voting can skew the results. Circuits with more peaked output probability distributions often benefit more from aggregation with plurality voting (see Methods).

\section{Discussion}
\label{sec:Implication}
To enhance the performance of present-day quantum computers, scientists and engineers devote considerable effort to finding and mitigating error sources.
However, device inaccuracies and computational errors tend to persist
even after heroic improvements. In particular, \emph{coherent errors} --- which often arise from unintentional mis-calibrations that may drift in time --- can significantly degrade performance (error mitigation techniques
run into limitations for incoherent errors, as proven in~\cite{Takagi2022} via lower bounds). 
Even without hardware improvement, our strategy boosts QC performance because systematic errors vary between certain symmetric implementations.
\emph{Symmetrization} is the process of creating variant implementations
of quantum computation on specific hardware, so as to diminish errors (Fig.~\ref{fig:pca}b) and improve QC performance.
In particular, we split a given number of executions of a quantum circuit into batches, and each batch is executed using a different realization that should, by symmetry, give the same outcome in the absence of inaccuracies.
To aggregate the measurement statistics of symmetrized runs, we show that appropriately chosen techniques produce strong gains on a commercial QC.

Aggregation by componentwise averaging is stable for measurement statistics of any type. We use it to demonstrate a
 2$\times$ fidelity improvement for QML and VQE algorithms which produce few low-frequency outputs. For the algorithms with many zero-frequency outputs (QFT-based adders, amplitude estimation, phase estimation, Monte Carlo sampling) where the output result is encoded in a small set of target output states, componentwise averaging gives little to no improvement since it cannot recover zero-frequency outputs. Plurality voting is stable for this type of measurement statistics and demonstrates an up to 100$\times$ performance boost on our 20-qubit commercial QC~\cite{IonQ_Q1_2022}.
Our error mitigation strategy appears applicable to multiple qubit technologies and is compatible with prior error-mitigation strategies. 

\section{Methods}
\label{sec:methods}

Here, we give additional details on the two steps of symmetrization: the sampling of symmetries and the aggregation of measurement statistics. We also outline several considerations of scalability for these two steps. Details on our experiment and simulation are given as well.

\subsection{Sampling symmetries}
\label{sec:samp}
Since using all possible symmetries for a given quantum computation is impractical, we need to sample from those symmetries. For an error-free quantum computation, it suffices to use the identity symmetry alone. Assuming a single inaccuracy of a known type, very few symmetries would be sufficient, regardless of the magnitude of inaccuracies or the number of qubits. As the dimensionality of the error space grows, more symmetries must be sampled.

We sample symmetries $\gamma$ to minimize $\left< \delta u_{\gamma} \right>$. Selecting dissimilar (rather than random~\cite{Takagi2022}) symmetries reduces the bias and decorrelates inaccuracies between the variants. If symmetries $\Gamma$ are qubit assignments, one may select assignments that share fewer gates between physical qubits for a given device-specific connectivity.

\subsection{Aggregation strategies}
\label{sec:aggr}
Continuing the discussion in Section II.B, 
we compare two aggregation strategies for measurement statistics:
one represents them by frequency distributions, and the other --- by 
raw output samples.

\noindent
{\bf Componentwise averaging.}
Our first strategy performs componentwise averaging of frequencies in given histograms~\cite{10.1145/3352460.3358257}.
It suites computations with few or no zeros in the ideal probability distribution, such as VQE or QML circuits.
Fig.~\ref{fig:pca}b represents with vectors the differences between the histogram of each variant and the ideal histogram. With an appropriate sampling of symmetries, these vectors cancel out and their sum converges to the ideal one as the number of variants increases.
Componentwise averaging is unable to recover zero frequencies in ideal output distributions.
Intuitively, averaging is related to the set-union operation, whereas set-intersection suggests different aggregation methods. Namely, methods based on voting and can filter out low-frequency outputs and recover zero frequencies.

\noindent
{\bf Plurality voting.}
To specify aggregation by plurality voting we represent measurement results for each circuit variant by a set of bitstrings, one per shot. Since each variant has the same number of shots, each shot can be represented by the same number of variant bitstrings (see Fig.~\ref{fig:flow}c). The winning bitstring is determined by the plurality vote that additionally exceeds a specified threshold. Since the order of bitstrings does not matter, voting per shot is repeated many times over the scrambled orderings of bitstrings in each variant. If no winning bitstring is found, the threshold is reduced by one.  If no winner exists for the threshold value of two, a componentwise average of variant histograms is returned (this is common for spread-out distributions and/or also when available samples lack statistical significance). After accumulating counts from all winning bitstrings,
the final histogram is normalized.
The voting threshold is determined by training runs for a given QC architecture. Executed for a set of circuits with known outputs, the training runs also help to determine optimal numbers of variants, repetitions, gate decompositions etc. These hyperparameters are used for multiple circuits.

Due to the nonlinearity of voting, it is a stable aggregation strategy for ideal output probabilities with $l$ equally probable outputs and $r-l$ zero frequencies
(see Supplemental Materials for proofs).
Relevant circuits include QFT-based adders, phase and amplitude estimation algorithms, some Monte Carlo algorithms.

\begin{figure}
\includegraphics[scale=0.32]{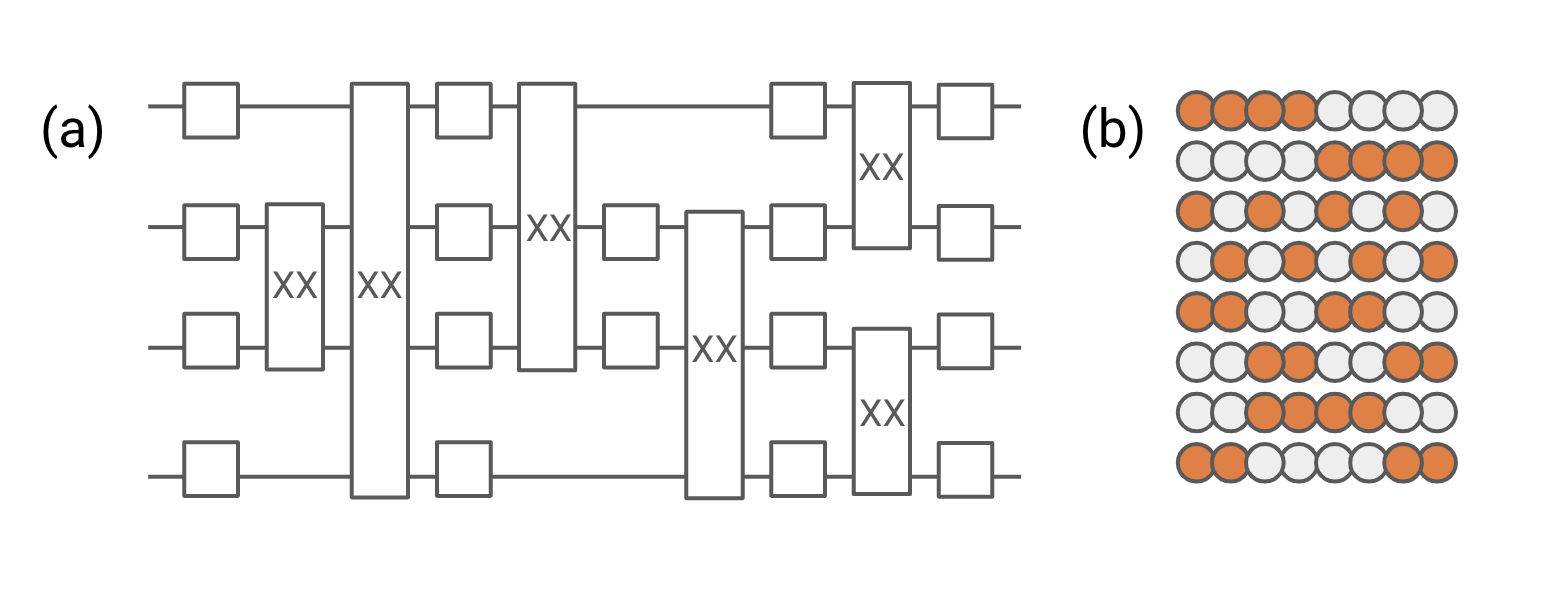}
\caption{\label{fig:sim}
Simulated 4-qubit small random circuit comprising XX gates and random single-qubit gates (a) was mapped onto eight qubits following the assignment highlighted by orange circles (b).
}
\end{figure}

\subsection{Considerations of scalability}

\noindent
{\bf Sampling of symmetries.} Uniformly random selection~\cite{Patel2020} offers a computationally scalable sampling method in that the memory of all previously selected symmetries is not necessary to select the next $\gamma$.  Since $k$ uniformly random symmetries produce a uniformly random set of $\delta u_\gamma$, we have $\left< \delta u_\gamma \right> \propto 1/\sqrt{k}$~\cite{Takagi2022}. Selecting dissimilar symmetries can reduce the expectation $\left< \delta u_\gamma \right>$, just like low-discrepancy sequences~\cite{Halton1960,Sobol1967,Fisher1974-un} improve upon random samples. To avoid specializing symmetry selection to each individual computation, we engineer it for entire classes of computation, possibly with moderate suboptimality. For example, similar VQE circuits (on the same number of qubits) can be viewed as one class. 

\noindent
{\bf Aggregation of measurement statistics.} Run time and memory complexity depend on the number of observed output states rather than the number of all possible states. For componentwise averaging, the postprocessing comes down to the simple or weighted merger of output counts (zero frequencies are implicit). Plurality voting is performed in small groups of outputs and does not require significant memory.

\subsection{Experimental details}
\label{sec:details}
We use the IonQ Aria~\cite{IonQ_Q1_2022} trapped-ion quantum computer which utilizes trapped Ytterbium ions individually addressed by pulses of 355 nm light. These pulses can be engineered to generate a M{\o}lmer-S{\o}rensen entangling gate between ions as well as single qubit rotations/gates. The Aria system uses 22 such ions as qubits to perform quantum information processing. Here, we split our experiments into 25 different maps (variants) between physical and computational qubits, running 100 experimental shots on each variant resulting in 2500 total experimental repetitions. For circuits on more than six qubits, we genrated permutations on a set of physical qubits. Otherwise, two additional physical qubits were utilized to increase the number of diverse mappings. All variants were measured under similar conditions. To this end, for most of our experiments, we executed our jobs within one calibration cycle. Whenever this was not possible (e.g. due to ion-chain loss), the calibration parameters were carefully replicated.

\subsection{Simulation details}

We show the effect of symmetrization on a 4-qubit random circuit (Fig.~\ref{fig:sim}a) in eight implementations with varying qubit assignment onto eight ions (Fig.~\ref{fig:sim}b). We model gate miscalibrations as random under-rotations of multiple two-qubit gates fixed per ion pair. We assume an average under-rotation across all qubit pairs causing a similar error for all variants (Fig.~\ref{fig:pca}a). Symmetries $\Gamma$ are represented by eight random qubit assignments $\gamma$. For each qubit assignment, we simulated corresponding inaccurate realizations to obtain vectors $\mathbf{h_u}+\delta h_
{\Gamma\text{-inv}}+\delta h_\gamma$.

In the first case (Fig.~\ref{fig:pca}a), we use only vary qubit assignment between the implementations (Fig.~\ref{fig:sim}b) while in the other case (Fig.~\ref{fig:pca}b), we also replace every fourth XX-gate with a phase-flipped XX$^{-1}$ gates with pairs of X-gates thus implementing the same ideal unitary but reversing the effect of under-rotation (Fig.~\ref{fig:pca}b).

\bibliography{citations}

\section{Acknowledgments}
\vspace{-2mm}
We thank John Gamble for insightful discussions and valuable suggestions.

\section{Contributions}
\vspace{-2mm}
I.M. and Y.N. conceived and coordinated the project. I.M. proposed the idea of the strategy and designed the methods with A.M. J.N. conducted the experiment, A.M. wrote and performed the simulations and data processing. All authors contributed to writing the manuscript.

\appendix

\section{Supplemental Information - validity and efficacy of plurality voting}

As detailed in the main text, plurality voting is a powerful aggregation strategy because it is \emph{nonlinear} and can strongly suppress errors for some circuits. 
However, it can also degrade performance if used for other circuits. 
Here, we formally analyze the properties of the plurality vote procedure, detailing the conditions that should be satisfied for its use to be beneficial.

Let us first consider the simple case with no finite-sample effects, no errors, and $r$ possible valid output states.
We consider $m$ variants with the probability $h_i$ to measure state $i$.
The probability to measure each output state a specified number of times $\{x_1,\dots,x_r\} = \mathbf{x}^r$ such that $\sum_{k=1}^r x_k = m$ can be written in terms of multinomial coefficients as
\begin{align}
    \gamma(m,\mathbf{x}^r)
    = \binom{m}{x_1,\dots,x_r} \prod_{j=1}^r h_j^{x_j},
\end{align}

We can then write down the probability to find a state $i$ exactly $x_i$ times out of $m$ variants by summing $\gamma(m,\mathbf{x}^r)$ over every variable in $\mathbf{x}^r$ except the prefixed $x_i$ denoting the constraint $\sum_{k=1}^r x_k = m$ with a primed sum as

\begin{align}
    \gamma_i(m,x_i)
    = \sideset{}{'} \sum_{\mathbf{x}^r\setminus x_i=0} \binom{m}{x_1,\dots,x_r} \prod_{j=1}^r h_j^{x_j},
\end{align}

The probability that the measured state $i$ is the most frequently measured state and is found at least $t$ times out of $m$ variants can be expressed as a sum over $\gamma(m,x_i)$ with an additional constraint that requires any $x_k \in \mathbf{x}^r \setminus x_i$ to be less than $x_i$:
\begin{align}
G_i(m,t) = \sum_{x_i=t}^m \phantom{x} \sideset{}{'} \sum_{\mathbf{x}^r\setminus x_i=0}^{x_i-1} \binom{m}{x_1,\dots,x_r} \prod_{j=1}^r h_j^{x_j}
\end{align}

The output probability of state $i$ in the aggregated results can be expressed through the normalized $G_i(m,t)$
\begin{align}
    g_i(m,t) = \frac{G_i(m,t)}{\sum_j G_j(m,t)}
\end{align}

\begin{theorem}
\label{theo:main}
For any ideal output probability distribution $\{h_1,\dots,h_r\}$ and any two states $1\leq i \neq j \leq r$, the corresponding aggregated output probabilities $g_i$, $g_j$ satisfy $g_i/g_j<h_i/h_j$ if $h_i<h_j \notin \{0,1\}$.
\end{theorem}
\begin{proof}
Let us consider an output probability distribution with $l$ nonzero output states. $G_i(m,t)$ can be written as
\begin{align}
    G_i(m,t) = \sum_{x_i=t}^m \sum_{x_j=0}^{\substack{\min(x_i-1,\\ m-x_i)}} h_i^{x_i} h_j^{x_j} f_{ij}^m(x_i,x_j),
\end{align}
where
\begin{align}
f_{ij}^m(x_i,x_j) = \sideset{}{'} \sum_{\mathbf{x}^r\setminus x_i\setminus x_j=0}^{x_i-1} \phantom{|} \prod_{k=1}^l h_k^{x_k} \binom{m}{x_1,\dots,x_l}
\end{align}
Let us change the summation over $t \leq x_i \leq m$ and $0 \leq x_j \leq \min(x_i-1,m-x_i)$ to $q=x_i+x_j$ and $u=x_i-x_j$ where $1 \leq u \leq m$ and $\max(u,2t-u) \leq q \leq m$, which can be confirmed geometrically, so that
\begin{multline}
    G_i(m,t) = \sum_{u=1}^{m} h_i^u \sum_{\substack{q=\max(u, \\ 2t-u)}}^m (h_i h_j)^\frac{q-u}{2}  f_{ij}^m(\tfrac{q+u}{2},\tfrac{q-u}{2}) = \\
    = \sum_{u=1}^{m} h_i^u \phi_{ij}(m,t,u),
\end{multline}
where $\phi_{ij}^m(t,u) = \sum_{q=\max(t,u)}^m (h_i h_j)^\frac{q-u}{2}  f_{ij}^m(\tfrac{q+u}{2},\tfrac{q-u}{2})$. The ratio between the aggregated output probabilities can be expressed as
\begin{multline}
    \frac{g_i(m,t)}{g_j(m,t)} = \frac{G_i(m,t)}{G_j(m,t)} = \frac{h_i}{h_j} \frac{\sum_{u=1}^{m} h_i^{u-1}\phi_{ij}^m(t,u)}{\sum_{u=1}^{m} h_j^{u-1}\phi_{ij}^m(t,u)}
\end{multline}
Comparing the sums term by term, since $h_i<h_j$, for $u \geq 1$, $h_i^{u-1} \leq h_j^{u-1}$ so that $\frac{g_i(m,t)}{g_j(m,t)}<\frac{h_i}{h_j}$.
\end{proof}

\begin{corollary}
If $h_i=0$, $g_i(m,t)=\alpha G_i(m,t) = \alpha \sum_{u=1}^{m} h_i^u \phi_{ij}(m,t,u) = 0$.
\end{corollary}

\begin{corollary}
If $h_i=1$, $h_{k\neq i} = g_{k \neq i} = 0$, $g_i=1$.
\end{corollary}

\begin{corollary}
For any output probability distribution $\{h_1,\dots,h_r\}$ such that $h_i=1/l$ for $1 \leq l \leq r$ states and $h_i=0$ for the rest, $g_i(m,t)=h_i$.
\begin{align}
    g_i(m,t) = \begin{cases}
			0, & h_i=0\\
            \frac{G_i(m,t)}{\sum^l G_j(m,t)} = \frac{G_i(m,t)}{l G_i(m,t)} = \frac{1}{l}, & h_i=1/l
		 \end{cases}
\end{align}
\end{corollary}

\begin{corollary}
If there is an imbalance $d$ between a state with probability $1/l$ and a state with probability $0$ so that $h_i=\frac{1}{l}-d$ and $h_j=d$, $\frac{g_i(m,t)}{g_j(m,t)}>\frac{h_i}{h_j}$ given that $h_i>h_j$ or that $0<d<\frac{1}{2l}$.
\end{corollary}

\begin{corollary}
If there is an imbalance $d$ between two states with probability $1/l$ so that $h_i=\frac{1}{l}-d$ and $h_j=\frac{1}{l}+d$, $\frac{g_i(m,t)}{g_j(m,t)}>\frac{h_i}{h_j}$ given that $h_i>h_j$ or that $0<d<\frac{1}{l}$.
\end{corollary}

It follows from Theorem~\ref{theo:main} and its corollaries that plurality voting is a stable aggregation strategy for ideal output probabilities with $1\leq l \leq r$ equally probable outputs and $r-l$ zero frequencies.
If the non-zero output probabilities differ, the smaller ones get further reduced in the aggregated results, while the larger ones get amplified. This property helps to reduce the aggregated probabilities of zero-frequency outputs when they are erroneously measured.

\end{document}